\RequirePackage{fix-cm}
\documentclass[smallextended]{svjour3}       % onecolumn (second format)
\usepackage{epstopdf}
\smartqed  % flush right qed marks, e.g. at end of proof
\usepackage{graphicx}
\usepackage{files/commands}
\usepackage{pdflscape}

\begin{document}

\title{Estimating Sequence Similarity from Read Sets for Clustering Next-Generation Sequencing data}

\titlerunning{Estimating Sequence Similarity from Read Sets}        % if too long for running head

\author{Petr Ry\v{s}av\'{y}         \and
        Filip \v{Z}elezn\'{y} %etc.
}


\institute{Petr Ry\v{s}av\'{y} \at
              Dept. of Computer Science, Faculty of Electrical Engineering \\ Czech Technical University in Prague, Czech Republic\\
              \email{petr.rysavy@fel.cvut.cz}           %  \\
           \and
           Filip \v{Z}elezn\'{y} \at
              Dept. of Computer Science, Faculty of Electrical Engineering \\ Czech Technical University in Prague, Czech Republic\\
              \email{zelezny@fel.cvut.cz}
}

\date{Received: date / Accepted: date}

\maketitle

\begin{abstract}
To cluster sequences given only their read-set representations, one may try to reconstruct each one from the corresponding read set, and then employ conventional (dis)similarity measures such as the edit distance on the assembled sequences. This approach is however problematic and we propose instead to estimate the similarities directly from the read sets.\\
Our approach is based on an adaptation of the Monge-Elkan similarity known from the field of databases. It avoids the NP-hard problem of sequence assembly. For low coverage data it results in a better approximation of the true sequence similarities and consequently in better clustering, in comparison to the first-assemble-then-cluster approach.
\end{abstract}

\section{Introduction}

{\em Sequencing} means reading the sequence of elements that constitute a polymer, such as the DNA.
The {\em human genome project} \citep{lander.nature.11} completed in 2003 was a prime example of sequencing, resulting in the identification of almost the entire genomic sequence (over 3 billion symbols) of a single human. 
Sequencing becomes technologically difficult as the length of the read sequence grows. The common principle of {\em new-generation sequencing} (NGS) is that only very short substrings (10's to 100's of symbols) are read at random positions of the sequence of interest. It is usually required that the number of such read substrings, called {\em reads}, is such that with high probability each position in the sequence is contained in multiple reads; the number of such reads is termed {\em coverage}. The complete sequence is determined by combinatorial assembly of the substrings guided by their suffix-prefix overlaps. For example, one possible assembly of reads $\{\str{AGGC}, \str{TGGA}, \str{GCT}\}$ is $\str{AGGCTGGA}$. Short reads imply low cost of wet-lab sequencing traded off with high computational cost of assembly. Indeed the assembly task can be posed as searching the Hamiltonian path in a graph of mutual overlaps.

One of the central tasks in computational biology is to infer phylogenetic trees, which typically amounts to hierarchical clustering of genomes. When they are represented only through sequences read-sets, the bioinformatician is forced to reconstruct the sequences from the read-sets prior to clustering. This of course entails the solution of the NP-hard assembly problem for each data instance with little guarantees regarding the quality of the resulting putative sequence. This motivates the question whether the assembly step could be entirely avoided. We address this question here by proposing a similarity function computable directly on the read sets, that should approximate the true similarities on the original sequences. 

In the next section we provide a brief overview of the related work and relevant methods. In Sect. \ref{sec:method}, we design the similarity (or, reversely distance) function. Then we provide a brief theoretical analysis of it. In Sect. \ref{sec:experiments} we compare it to the conventional approach on genomic data and then we conclude the paper.

\section{Related Work}

Related work includes studies on clustering NGS data (e.g. \cite{seed,pipe,kchouk2016clustering}). They however deal with clustering {\em reads} and we are not aware of a previous attempt to cluster {\em read-sets}. The paper  \citep{distance} proposes a similarity measure for NGS data, but again it operates on the level of reads. The previous work \citep{zelezny.14} and \citep{jalovec.14} also aims at avoiding the assembly step in learning from NGS data but these studies concern supervised classification learning and they do not elaborate on read-set similarity.

This work extends \citep{ida2016}, which corresponds to Sects. \ref{ssec:mongeelkan}, \ref{ssec:scale}, \ref{ssec:weakborder} and \ref{ssec:threshold}. However this work fixes several issues that made the original approach hard to apply on real data. In this paper we also provide a more detailed experimental evaluation.

\emph{Approximate string matching} area is related to the method we provide. Here the goal is to find strings similar to a pattern in a dictionary. Namely we employ the $q$-gram distance \citep{ukkonenQgram}. Our method also includes the Monge-Elkan distance \citep{mongeelkan} known from the field of databases.

\section{Proposed method}
\label{sec:method}

The functor $|.|$ will denote the absolute value, cardinality and length (respectively) for a number, set, and string argument. Let $\dist(A,B)$ denote the Levenshtein distance \citep{levenstein} between strings $A$ and $B$. The function measures the minimum number of edits (insertions, deletions, and substitutions) needed to make the strings identical, and is a typical example of a sequence dissimilarity measure used in bioinformatics. It is a property of the distance that
\begin{equation}\label{eq:levenupper}
  \dist(A,B) \leq \max\{|A|,|B|\} .
\end{equation}

We will work with constants $l \in \mathbb{N}, \alpha \in \mathbb{R}^{+}$ called the {\em read length} and {\em coverage}, respectively, which are specific to a particular sequencing experiment. A {\em read-set} $R_A$ of string $A$ such that \begin{equation}\label{eq:Alarge}                                                                                                                                                                                                                      |A| \gg l                                                                                                                                                                                                                       \end{equation}
is a multiset of\footnote{Should the right hand side be non-integer, we neglect its fractional part.}
\begin{equation}\label{eq:coverage}
|R_A| = \frac{\alpha}{l}|A|  
\end{equation}
 substrings sampled i.i.d. with replacement from the uniform distribution on all the $|A|-l+1$ substrings of length $l$ of $A$.  Informally, the coverage $\alpha$ indicates the average number of reads covering a given place in $A$.
 
Throughout the paper we hold assumption that the coverage $\alpha$ and read length $l$ are constant. The constant read length assumption can be justified by the way how some of the sequencing machines work. They read one nucleotide of each read every iteration.

Our goal is to propose a distance function $\setdist(R_A,R_B)$ that approximates $\dist(A,B)$ for read-sets $R_A$ and $R_B$ of arbitrary  strings $A$ and $B$.  We also want $\setdist(R_A,R_B)$ to be more accurate and less complex to calculate than a natural estimate $\dist(\hat{A},\hat{B})$ in which the arguments represent putative sequences reconstructed from $R_A$ and $R_B$ using {\em assembly algorithms} such as \citep{abyss,edena,ssake,spades,velvet}.

Alphabet that we use in our model is $\Sigma = \{ \str{A}, \str{T}, \str{C}, \str{G} \}$. If $a_i$ is a read, then $\overline{a}_i$ denotes its complement, where $\str{A}$ (resp. $\str{C}$) and $\str{T}$ (resp. $\str{G}$) are interchanged. For example complement of read $\str{ATTCG}$ is read $\str{TAAGC}$. Similarly we define reversed read, denoted $\reverse{a_i}$. For example reverse of $\str{ATTCG}$ is $\str{GCTTA}$.

\subsection{Base Case: Which Reads Belong Together}
\label{ssec:mongeelkan}

A natural approach to instantiate $\setdist(R_A,R_B)$ is to exploit the $|R_A||R_B|$ pairwise Levenshtein distances between the reads in $R_A$ and $R_B$. Most of those values are useless because they match reads from completely different parts of sequences $A$ and $B$. Therefore we want to account only for those pairs which likely belong together.

If we seek a read from $R_B$ that matches a read $a_i \in R_A$, we make the assumption that the most similar read $b_j \in R_B$ is the one that we look for (see Fig. 
\ref{fig:mongeelkan}), i.e.,
\begin{equation*}
  b_j = \argmin_{b_{k} \in R_B} \dist (a_i, b_{k}).
\end{equation*}
To calculate the distance from $R_A$ to $R_B$, we average over all reads from $R_A$:
\begin{equation}  
  \setdist_{\textsf{ME}} (R_A, R_B) 
   = \frac{1}{|R_A|} \sum_{a_i \in R_A} \min_{b_j \in R_B} \dist (a_i, b_j).
   \label{eq:mongeelkan}
\end{equation}
This idea was presented in \citep{mongeelkan} for searching duplicates in database systems. The method is known as the {\em Monge-Elkan similarity}\footnote{Here we alter the Monge-Elkan similarity into a distance measure. The standard way of using Monge-Elkan is as a similarity measure with $\min$ replaced by $\max$ and distance calculation by similarity calculation.} (hence the {\scriptsize \textsf{ME}} label) and entails a simple but effective approximation algorithm.

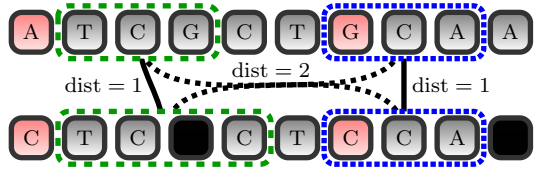
\begin{figure}[t]
  \centering
  \begin{tikzpicture}[
    y=2cm,
    letter/.style = {draw=black!80, fill=black!10, shade, line width=2pt, inner sep=1pt,minimum size=15pt, rounded corners},
    miss/.style = {bottom color=red!10, top color=red!50},
    gap/.style = {bottom color=black, top color=black},
    read/.style = {line width=2pt, rounded corners, fill=white},
    scale=0.7,
]

\draw[line width=2pt, dotted] (2,1) to[in = 90, out=-90] (7,0);
\draw[line width=2pt, dotted] (7,1) to[in = 90, out=-90] (2.5,0);

\node (X) at (4.5,0.5) [above] {$\mathrm{dist} = 2$};

\draw[line width=2pt] (2,1) to[in=90, out=-90] (2.5,0);
\draw[line width=2pt] (7,1) to[in=90, out=-90] (7,0);
\node (Y) at (2.25,0.5) [left] {$\mathrm{dist} = 1$};
\node (Z) at (7,0.5) [right] {$\mathrm{dist} = 1$};

\draw[read, draw=green!60!black, dashed] (0.5,1.25) rectangle +(3,-0.5);
\draw[read, draw=green!60!black, dashed] (0.5,0.25) rectangle +(4,-0.5);

\draw[read, draw=blue, densely dotted] (5.5,1.25) rectangle +(3,-0.5);
\draw[read, draw=blue, densely dotted] (5.5,0.25) rectangle +(3,-0.5);

\node[letter, miss] (A0) at(0,1) {A};
\node[letter] (A1) at(1,1) {T};
\node[letter] (A2) at(2,1) {C};
\node[letter] (A3) at(3,1) {G};
\node[letter] (A4) at(4,1) {C};
\node[letter] (A5) at(5,1) {T};
\node[letter, miss] (A6) at(6,1) {G};
\node[letter] (A7) at(7,1) {C};
\node[letter] (A8) at(8,1) {A};
\node[letter] (A9) at(9,1) {A};

\node[letter, miss] (B0) at(0,0) {C};
\node[letter] (B1) at(1,0) {T};
\node[letter] (B2) at(2,0) {C};
\node[letter, gap] (B3) at(3,0) {};
\node[letter] (B4) at(4,0) {C};
\node[letter] (B5) at(5,0) {T};
\node[letter, miss] (B6) at(6,0) {C};
\node[letter] (B7) at(7,0) {C};
\node[letter] (B8) at(8,0) {A};
\node[letter, gap] (B9) at(9,0) {};
\end{tikzpicture}
  \caption{We calculate read-read distances in order to find matching pairs of reads. For each read from the first sequence we find the least distant read in the second sequence. We see optimal alignment of \str{ATCGCTGCAA} and \str{CTCCTCCA}. Read \str{TCG} is paired with \str{TCC}.}
  \label{fig:mongeelkan}
\end{figure}

$\setdist_{\textsf{ME}}(R_A,R_B)$ is non-symmetric in general, which is undesirable given that the approximated distance $\dist (A,B)$ is known to be symmetric. Therefore we define a symmetric version by averaging both directions
\begin{equation}
 \setdist_{\textsf{MES}} (R_A, R_B)
   = \frac{1}{2} \left( \setdist_{\textsf{ME}} (R_A, R_B) + \setdist_{\textsf{ME}} (R_B, R_A) \right).
   \label{eq:mongeelkansym}
\end{equation}

\subsection{Strand and orientation}

In practical setting we do not know which DNA strand do the reads come from.
If we match read $a_i$ with read $b_j$, there are two possible matchings.
If reads come from the complementary strands, we need to calculate
$\dist(a_i, \reverse{\overline{b}_j})$ besides $\dist(a_i, b_j)$.
We consider only the option that leads to lower distance.

Some sequencing techniques lead to loss of information which end of read belongs to
3'-end and which end belongs to 5'-end. In that case we have two other independent
options. Reads $a_i$ and $b_j$ either match or their reverses match. We calculate
both distances and the orientation that leads to lower distance is used as the most
likely one.

Based on the sequencing setting we have up to four options how to match reads
$a_i$ and $b_j$. The options that need to be considered are subset of
$\dist(a_i, b_j)$, $\dist(a_i, \overline{b}_j)$, $\dist(a_i, \reverse{b_j})$
and $\dist(a_i, \reverse{\overline{b}_j})$.

\subsection{Distance Scale}
\label{ssec:scale}
Consider duplicating a non-empty string $A$ into $AA$ and assume $R_{AA} = R_A 
\cup R_A$. Typically for a $B$ similar to $A$ we expect that $\dist(AA,B) > \dist(A,B)$ but the (symmetric) Monge-Elkan distance will not change, i.e. $\setdist_{\textsf{MES}} (R_{AA}, R_B) =  \setdist_{\textsf{MES}} (R_{A}, R_B)$, indicating a discrepancy that should be rectified.

In fact, $\setdist_{\textsf{MES}}$ has the constant upper bound $l$, which is because it is the average (c.f. \eqref{eq:mongeelkan} and \eqref{eq:mongeelkansym}) of numbers no greater than $l$ (see \eqref{eq:levenupper}). On the other hand, $\dist(A,B)$ has a non-constant upper bound $\max\{|A|,|B|\}$ as by \eqref{eq:levenupper}.

To bring $\setdist_{\textsf{MES}}(A,B)$ on the same scale as $\dist(A,B)$, we should therefore multiply it by the factor $\max\{|A|,|B|\}/l = \max\{|A|/l,|B|/l\}$. By \eqref{eq:coverage} we have $|A| = \frac{l}{\alpha}|R_A|$, yielding the factor $\max\{|R_A|/\alpha,|R_B|/\alpha\}$, in which $\alpha$ is a constant divisor which can be neglected in a distance function. Therefore, we modify the read distance into
\begin{equation*}
\setdist_{\textsf{MESS}}(R_A,R_B) =  \max\{|R_A|,|R_B|\}  \setdist_{\textsf{MES}}(R_A,R_B).
\end{equation*}

\subsection{Margin Gaps}
\label{ssec:weakborder}
Consider the situation in Fig. \ref{fig:weakborder} showing two identical sequences each with one shown read. The Levenshtein distance between the two reads is non-zero due to the one-symbol trailing (leading, respectively) gap of the top (bottom) read caused only by the different random positions of the reads rather than due to a mismatch between the sequences. Thus there is an intuitive reason to pardon margin gaps up to certain size $t$ \begin{equation}\label{eq:Tsmall} t < \frac{l}{2}\end{equation}
when matching reads. Here, $t$ should not be too large as otherwise the distance could be nullified for pairs of long reads with small prefix-suffix overlaps, which would not make sense.

\begin{figure}[t]
  \centering
  \usetikzlibrary{patterns,snakes}

\begin{tikzpicture}[
    y=2cm,
    letter/.style = {draw=black!80, fill=black!10, shade, line width=2pt, inner sep=1pt,minimum size=15pt, rounded corners},
    miss/.style = {bottom color=red!10, top color=red!50},
    gap/.style = {bottom color=black, top color=black},
    read/.style = {line width=2pt, rounded corners, fill=white},
    scale=0.7,
]

\draw[line width=2pt] (2.5,1) to[in=90, out=-90] (3.5,0);
\node (Y) at (3,0.5) [right] {$\mathrm{dist} = 0$};

\draw[read, draw=green!60!black, dashed] (0.5,1.25) rectangle +(4,-0.5);
\draw[read, draw=green!60!black, dashed] (1.5,0.25) rectangle +(4,-0.5);

\draw[draw=green!60!black, dotted, line width=1pt] (1.5, 1.25) -- +(0,-0.5);
\draw[draw=green!60!black, dotted, line width=1pt] (3.5, 1.25) -- +(0,-0.5);
\draw[draw=green!60!black, dotted, line width=1pt] (2.5, 0.25) -- +(0,-0.5);
\draw[draw=green!60!black, dotted, line width=1pt] (4.5, 0.25) -- +(0,-0.5);

\draw [
    thick,
    decoration={brace, mirror, raise=3pt},
    decorate
] (0.5,0.75) -- +(1,0)
node[pos=0.5, below=3pt]{$t=1$};

\node[letter] (A0) at(0,1) {A};
\node[letter] (A1) at(1,1) {T};
\node[letter] (A2) at(2,1) {C};
\node[letter] (A3) at(3,1) {G};
\node[letter] (A4) at(4,1) {C};
\node[letter] (A5) at(5,1) {T};
\node[letter] (A6) at(6,1) {G};

\node[letter] (B0) at(0,0) {A};
\node[letter] (B1) at(1,0) {T};
\node[letter] (B2) at(2,0) {C};
\node[letter] (B3) at(3,0) {G};
\node[letter] (B4) at(4,0) {C};
\node[letter] (B5) at(5,0) {T};
\node[letter] (B6) at(6,0) {G};

\end{tikzpicture}
  \caption{Because reads locations in sequences are random, we do not want to penalize small leading or trailing gaps.}
  \label{fig:weakborder}
\end{figure}
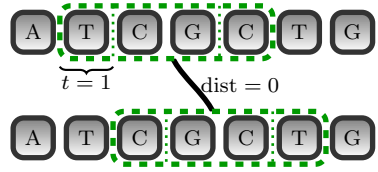

To estimate a good value for $t$, consider sequence $A$ and its sampled read-set $R_A$. We now sample an additional read $a$ of length $l$ from $A$. Ideally, there should be a zero-penalty match for $a$ in $R_A$ as $a$ was sampled from the same sequence as $R_A$ was. This happens iff there is a read in $R_A$ sampled from the same position in $A$ as $a$, or from a position shifted by up to $t$ symbols to the left or right as then the induced gaps are penalty-free. Since $R_A$ is an uniform-probability i.i.d. sample from $A$, the probability that a particular read from $R_A$ starts at one of these $1+2t$ positions is\footnote{Strictly speaking, this reasoning is incorrect if read $a$ is drawn from a place close to $A$'s margins, more precisely, if it starts in fewer than $t$ ($t+l$, respectively) symbols from $A$'s  left (right) margin, as then not all of the $2t$ shifts are possible. This is however negligible due to Ineq. \eqref{eq:Alarge}.}  $\frac{1+2t}{|A|}$. We want to put an upper bound $\varepsilon>0$ on the probability that this happens for none of the $|R_A|=\frac{\alpha}{l}
|A|$ reads in $R_A$:
\begin{equation*}
  p = \left(
    1- \frac{1+2t}{|A|}
  \right)^{\frac{|A|\cdot \alpha}{l}} \leq \varepsilon.
  \label{eq:probability}
\end{equation*}
Consider the first-order Taylor approximation 
$
(1+x)^n = 1+nx + \varepsilon'
$ where the difference term
 $\varepsilon'>0$ decreases with decreasing $|x|$. 
Due to Ineq. \eqref{eq:Tsmall} and \eqref{eq:Alarge}, $\frac{1+2t}{|A|}$ is small and we can apply the approximation on the above formula for $p$, yielding
\begin{equation*}
 p =  1 - \frac{2t+1}{|A|}\frac{|A|\cdot \alpha}{l} + 
  \varepsilon'= 
  1 - (2t + 1) \frac{\alpha}{l}  + 
  \varepsilon'\leq 
  \varepsilon.
\end{equation*}
For simplicity, we choose $\varepsilon = \varepsilon'$. The smallest gap size $t$ for which the inequality is satisfied is obtained by solving $1 - (2t + 1) \frac{\alpha}{l} = 0$, yielding
\begin{equation}
  t = \frac{1}{2} \left( \frac{l}{\alpha} - 1 \right).
  \label{eq:tvalue}
\end{equation}
This choice of $t$ matches intuition in that with larger read-length $l$ we can allow a larger grace gap $t$ but with larger coverage $\alpha$, $t$ needs not be so large as there is a higher chance of having a suitably positioned read in the read-set. Another way to look at it is to realize that 
reads in a read-set are approximately $\frac{l}{\alpha}$ positions from each other. Consider matching read $a$ to reads from $R_A$. If there is a read $a_1 \in R_A$ requiring gap larger than $\frac{l}{2\alpha}$ to match $a$, then there will typically be another read $a_2 \in R_A$ requiring gap at most $\frac{l}{2\alpha}$ (see Fig. 
\ref{fig:coverage}). From \eqref{eq:Tsmall} and \eqref{eq:tvalue} we see that this method is applicable only when $\alpha > (\frac{1}{l}+1)^{-1}$, which is bit less than $1$. However the results start to be nonzero for $\alpha > (\frac{1}{l}+2)^{-1}$, which is bit less than $0.5$.

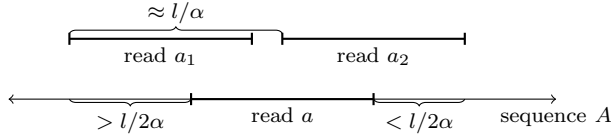
\begin{figure}
\begin{center}
\begin{tikzpicture}[scale=0.8]
\draw[<->] (0,0) -- (9,0) node[below] {sequence $A$};
\draw[thick] (3,0) -- (6,0) node[below, midway] {read $a$}; \draw[thick] (3,-0.1) -- (3,0.1);  \draw[thick] (6,-0.1) -- (6,0.1);
\draw[thick] (1,1) -- (4,1) node[below, midway] {read $a_1$};\draw[thick] (1,0.9) -- (1,1.1);  \draw[thick] (4,0.9) -- (4,1.1);
\draw[thick] (4.5,1) -- (7.5,1) node[below, midway] {read $a_2$};\draw[thick] (4.5,0.9) -- (4.5,1.1);  \draw[thick] (7.5,0.9) -- (7.5,1.1);
\draw[decoration={brace,mirror},decorate] (1,0) -- (3,0) node[below=2pt, midway] {$>l/2\alpha$};
\draw[decoration={brace,mirror},decorate] (6,0) -- (7.5,0) node[below=2pt, midway] {$<l/2\alpha$};
\draw[decoration={brace,raise=2pt},decorate] (1,1) -- (4.5,1) node[above=3pt,midway] {$\approx l/\alpha$};
\end{tikzpicture}
\end{center}
\caption{Illustration to reasoning in Sect. \ref{ssec:weakborder}}\label{fig:coverage}
\end{figure}

We implemented the grace margin gaps into a further version $\setdist_{\textsf{MESSG}}$ of the constructed distance function, which required only a small change to the standard Wagner-Fischer algorithm \citep{wagnerfischer}.\footnote{The dynamic programming algorithm for calculating the Levenshtein distance \citep{levenstein} is commonly called Wagner-Fischer algorithm \citep{wagnerfischer}. When we refer to sequence alignment problem in bioinformatics, this algorithm is often called Needleman-Wunsch algorithm \citep{needlemanwunsch}.} When the algorithm is filling the first or the last row and column of the table, margin  gaps up to $t$ symbols are not penalized. Larger margin gaps are penalized in a way that satisfies the constraint  that the distance between a word $a$ and an empty word is $|a|$. In particular, the standard linear gap penalty is replaced with a piecewise linear function that gives cost of margin gap at $x$-th position
 
\begin{equation}
  g(x) = \begin{cases}
    0, & \mbox{if } 0 \leq x \leq t - 1, \\
    2\frac{x-t + 1}{l+1-2t}, & \mbox{if } t-1 < x \leq l-t, \\
    2, & \mbox{if } l-t < x < l.
  \end{cases}
  \label{eq:weakborder}
\end{equation}

\subsection{Missing Read}
\label{ssec:threshold}

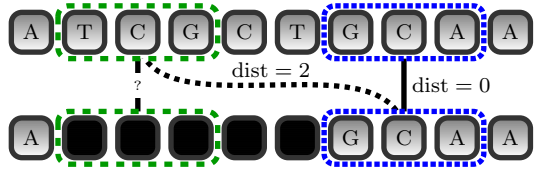
\begin{figure}
  \centering
  \begin{tikzpicture}[
    y=2cm,
    letter/.style = {draw=black!80, fill=black!10, shade, line width=2pt, inner sep=1pt,minimum size=15pt, rounded corners},
    miss/.style = {bottom color=red!10, top color=red!50},
    gap/.style = {bottom color=black, top color=black},
    read/.style = {line width=2pt, rounded corners, fill=white},
    scale=0.7,
]

\draw[line width=2pt, dotted] (2,1) to[in = 90, out=-90] (7,0);
%\draw[line width=2pt, dotted] (7,1) to[in = 90, out=-90] (2.5,0);

\node (X) at (4.5,0.5) [above] {$\mathrm{dist} = 2$};

\draw[line width=2pt] (2,1) -- (2,0) node[pos=0.5, fill=white] {\tiny $?$};
\draw[line width=2pt] (7,1) to[in=90, out=-90] (7,0);
%\node (Y) at (2.25,0.5) [left] {$\mathrm{dist} = 1$};
\node (Z) at (7,0.5) [right] {$\mathrm{dist} = 0$};

\draw[read, draw=green!60!black, dashed] (0.5,1.25) rectangle +(3,-0.5);
\draw[read, draw=green!60!black, dashed] (0.5,0.25) rectangle +(3,-0.5);

\draw[read, draw=blue, densely dotted] (5.5,1.25) rectangle +(3,-0.5);
\draw[read, draw=blue, densely dotted] (5.5,0.25) rectangle +(3,-0.5);

\node[letter] (A0) at(0,1) {A};
\node[letter] (A1) at(1,1) {T};
\node[letter] (A2) at(2,1) {C};
\node[letter] (A3) at(3,1) {G};
\node[letter] (A4) at(4,1) {C};
\node[letter] (A5) at(5,1) {T};
\node[letter] (A6) at(6,1) {G};
\node[letter] (A7) at(7,1) {C};
\node[letter] (A8) at(8,1) {A};
\node[letter] (A9) at(9,1) {A};

\node[letter] (B0) at(0,0) {A};
\node[letter, gap] (B1) at(1,0) {};
\node[letter, gap] (B2) at(2,0) {};
\node[letter, gap] (B3) at(3,0) {};
\node[letter, gap] (B4) at(4,0) {};
\node[letter, gap] (B5) at(5,0) {};
\node[letter] (B6) at(6,0) {G};
\node[letter] (B7) at(7,0) {C};
\node[letter] (B8) at(8,0) {A};
\node[letter] (B9) at(9,0) {A};
\end{tikzpicture}
  \caption{If the distance between a read and its closest counterpart is greater than threshold $\theta$, we assume that the read matches to a gap in the sequence alignment.}
  \label{fig:threshold}
\end{figure}

Sometimes there is no good match for read $a_i$ in $R_B$. During evolution the substring that contained $a_i$ may have been inserted into $A$ or may have vanished from $B$. Therefore if
\begin{equation*}
  \dist(a_i, b_j) \geq \theta
\end{equation*}
for some reads $a_i$ and $b_j$ and threshold $\theta$, we consider $a_i$ and $b_j$ to be dissimilar and we force their distance to be $l$. (See Fig. \ref{fig:threshold}.)

Threshold $\theta$ should be a linear factor of the maximal distance between two sequences of length $l$, i.e. $\theta = \theta' \cdot l$. Value of $\theta'$ should reflect the probability that the read is in one sequence and not in the other. Because the true probability is hidden, it needs to be determined empirically.

The distance function equipped with the missing read detection as described gives rise to the version denoted as $\setdist_{\textsf{MESSGM}}$.

\subsection{Sampling}
\label{ssec:sampling}

Unlike the assembly, the approach in this paper does not require high coverage
to produce good quality data. Therefore for data sequenced with high coverage
we can randomly sample only a small amount of reads to improve runtime. From
the experiments we will see that coverage between $2$ and $3$ provides a good
balance between runtime and quality of results.

\subsection{Embedding}
\label{ssec:embedding}

Another approach that can further reduce the running time is based on using
embedding for finding a good candidate for the best match. The minimum
in \eqref{eq:mongeelkan} does not need to be exact, but may be approximate.
We therefore embedd reads to a space of lower dimension, where we can calculate
the distance faster. Then we find reads minimizing the embedded distance.
For this small set of reads we evaluate the true Levenshtein distance with
grace margin gap penalty $\setdist_{\textsf{MESSG}}$ to find better estimate for use in \eqref{eq:mongeelkan}.

The embedding that we use is based on $q$-gram profile and $q$-gram distance
as formalized in \citep{ukkonenQgram}. $q$-gram is any string from $\Sigma^q$.
Then \emph{$q$-gram profile} of string $a$ is a vector
$\vec{Q}(a) \in \mathbb{N}_0^{|\Sigma|^{q}}$ that contains number or occurences
for all possible $q$-grams in $a$. Finally \emph{$q$-gram distance} $\dist_q(a,b)$ of strings
$a$ and $b$ is Manhattan distance of their $q$-gram profiles, i.e.
$ \dist_q(a,b) = \| \vec{Q}(a) - \vec{Q}(b) \|_1 $.

We use $q$-gram distance to approximate the Levenshtein distance to find
candidates that minimize the right hand side of \eqref{eq:mongeelkan}. The
$q$-gram distance forms a lower bound on edit distance (see Theorem \ref{thm:lb}). 
Similarly to BLAST \citep{blast} algorithm we use $q=3$, which provides
a good balance between the runtime and quality of estimates. The dimension
of embedded space of reads is equal to $|\Sigma|^q$, which is $64$ in our case.
In the next sections we will refer to this modification of distance as
$\setdist_\str{MESSGq}$

\begin{theorem} \label{thm:lb}
  For any reads $a$ and $b$ and $q=3$ holds
  $$ \dist(a,b) \geq \frac{1}{6} \textstyle\dist_q(a,b). $$
\end{theorem}
\begin{proof}
  We will prove the theorem by mathematical induction on Levenshtein distance between $a$ and $b$. Suppose that $\dist(a,b)=k$.
  \begin{itemize}
    \item If $k=0$, then both $\dist(a,b)=\dist_q(a,b)=0$ and theorem holds.
    \item Suppose that theorem holds for any reads with distance at most $k-1$. There exists a sequence of $k$ operations \emph{insert}, \emph{delete} and \emph{replace} that transforms $a$ into $b$. Denote $b'$ the result that we obtain after applying the first $k-1$ operations. Then $\dist(a,b')=\dist(a,b)-1$. Consider the last operation and calculate the maximal value of $\dist_q(a,b)-\dist_q(a,b')$.
    
    In case of insertion (or vice versa deletion) we obtain in worst case instead of $q$-grams $abc$,$bcd$ $q$-grams $abX$,$bXc$,$Xcd$ (or vice versa). The $q$-gram distance grows by at most $5$. In case of mismatch $q$-grams $abX$,$bXc$,$Xcd$ are replaced by $q$-grams $abY$,$bYc$,$Ycd$. The $q$-gram distance grows at most by $6$. Therefore $\dist_q(a,b)-\dist_q(a,b') \leq 6$. Putting this together with the induction assumption leads to
    \begin{align*}
     \dist(a,b) &= 1 + \dist(a,b')
                  \geq 1 + \frac{1}{6} \textstyle \dist_q(a,b')
                  \\ &
                  \geq 1 + \frac{1}{6} \left(\textstyle \dist_q(a,b)-6\right)
                  = \frac{1}{6} \textstyle \dist_q(a,b).
    \end{align*}
  \end{itemize}
  By mathematical induction the theorem holds. \qed
\end{proof}

\section{Theoretical Analysis}

\subsection{Asymptotic Complexity}
\label{ssec:asymptotic}

Calculating $\dist(A,B)$ for sequences $A$ and $B$ requires $\Theta(|A||B|)$ operations if we use the standard Wagner-Fischer dynamic programming algorithm \citep{wagnerfischer}. This algorithm also requires $\Theta(\min(|A|,|B|))$ memory as we are interested only in distance and not in the alignment. To calculate $\setdist_{\textsf{ME}}$ we need to know the distances between all pairs of reads, so we have to evaluate (see \eqref{eq:coverage}) $\frac{\alpha}{l}|A|\frac{\alpha}{l}|B|$ distances where each one requires $l^2$ operations. Therefore $\alpha^2 |A||B|$ operations are required. For the symmetric version $\setdist_{\textsf{MES}}$ we make $2\alpha^2 |A||B|$ operations, which can be reduced to $\alpha^2 |A||B|$ operations and $\Theta(l + \frac{\alpha}{l} (|A|+|B|))$ memory. Further modifications ({\scriptsize \textsf{MESS}, \textsf{MESSG}, \textsf{MESSGM}}) do not change the asymptotic complexity.

Sampling, as described in Sect. \ref{ssec:sampling} reduces runtime by $\alpha^{2}$ factor to $\Theta(|A||B|)$, assuming that the final coverage is a constant. Method {\scriptsize \textsf{MESSGq}} does not give any theoretical guarantee on how many pairwise edit distances are explored. However if we assume that this number is a small constant,  we get runtime of $\Theta(|R_A||R_B|+l^2(|R_A|+|R_B))$.

The constants $\alpha$ and $l$ are determined by the sequencing technology and the independent complexity factors are $|A|$ and $|B|$. To calculate the distance in the conventional way as $\dist(\hat{A}, \hat{B})$ requires to reconstruct $\hat{A}$ and $\hat{B}$ from the respective read-sets through an assembly algorithm. This is an NP-hard problem which becomes non-tractable for large  $|A|$ and $|B|$, and which is avoided by our approach.

\subsection{Metric Properties}
$\setdist_{\textsf{MES}}$ as well as the later versions are all symmetric and non-negative but none of the proposed versions satisfies the identity condition ($\dist(a,b)=0$ iff $a=b$) or the triangle inequality, despite being based on the Levenshtein distance $\dist$, which is a metric. For example, let $R_A=\{\str{ATC},\str{ATC},\str{GGG}\}$, let $R_B=\{\str{ATA},\str{GGG}\}$, and let $R_C=\{\str{CTA}, \str{GGG}\}$.
Then $\setdist_{\textsf{MES}}(R_A,R_B) = \frac{7}{12}$,
and $\setdist_{\textsf{MES}}(R_B,R_C) = \frac{1}{2}$
but $\setdist_{\textsf{MES}}(R_A, R_C) = \frac{14}{12} > \frac{7}{12}+ \frac{1}{2}$.
While this might lead to counter-intuitive behavior of the proposed distances in certain applications, the violated conditions are not requirements assumed by clustering algorithms.

\section{Experimental Evaluation}
\label{sec:experiments}

% the pictures
\begin{table*}
  \caption{Overview of the datasets used in the experiments.}
  \label{tab:datasets}
  \centering
  \begin{tabular}{lllllllll}
    \hline
    Name		& Source			& Read generation			& Strand known	& 5' to 3' known	& $n$	& $\alpha$			& $l$				& Time-limit \\
    \hline
    Influenza	& ENA				& i.i.d., uniform distr.	& \no			& \no				& $13$	& $0.1$ to $100$	& $3$ to $500$		& $2$ hours \\
    Various		& ENA				& i.i.d., uniform distr.	& \no			& \no				& $18$	& $0.1$ to $100$	& $3$ to $500$ 		& $2$ hours \\
    Hepatitis	& ENA				& \citep{art}				& \no			& \yes				& $81$	& $\{10, 30, 50\}$	& $\{30, 70, 100\}$	& $1$ day \\
    Chromosomes	& \citep{1kgenomes}	& real-world				& \no			& \yes				& $23$	& $4.32$			& $76$				& $1$ day \\
    \hline
  \end{tabular}
\end{table*}

\begin{landscape}
\begin{table*}[t]
  \centering
  \caption{Runtime, Pearson's correlation coefficient between distance matrices and Fowlkes-Mallows index for $k=4$ and $k=8$. The `reference' method calculates distances from the original sequences. We show only two assembly algorithms that gave the highest correlation and the better algorithm of pairs {\scriptsize \textsf{MES}}/{\scriptsize \textsf{MESS}} and {\scriptsize \textsf{MESSG}}/{\scriptsize \textsf{MESSGM}}.}
\label{tab:results}
  \pgfplotstabletypeset[
    create on use/finishedst/.style={create col/assign/.code={
        \edef\entry{\thisrow{finished}/\thisrow{outof}}
        \pgfkeyslet{/pgfplots/table/create col/next content}{\entry}}},
    columns={empty,visibleName,finishedst,ms/assem,ms/matrix, ms/upgma, ms/nj, correlation, FM/upgmak4, FM/upgmak8, FM/njk4, FM/njk8},
    columns/empty/.style={column name=Dataset,column type=l,string type},
    columns/visibleName/.style={string type, column name={method}},
    columns/finishedst/.style={string type, column name={finished}},
    columns/ms/assem/.style={fixed,column name={$\frac{\mbox{assem.}}{\mathrm{ms}}$}},
    columns/ms/matrix/.style={fixed,column name={$\frac{\mbox{distances}}{\mathrm{ms}}$}},
    columns/ms/upgma/.style={column name={$\frac{\mbox{UPGMA}}{\mathrm{ms}}$}},
    columns/ms/nj/.style={column name={$\frac{\mbox{NJ}}{\mathrm{ms}}$}},
    columns/correlation/.style={precision=3, skip 0., column name={corr.}},
    columns/FM/upgmak4/.style={skip 0., column name={$\frac{\mbox{UPGMA}}{B_4}$}},
    columns/FM/upgmak8/.style={skip 0., column name={$\frac{\mbox{UPGMA}}{B_8}$}},
    columns/FM/njk4/.style={skip 0., column name={$\frac{\mbox{NJ}}{B_4}$}},
    columns/FM/njk8/.style={skip 0., column name={$\frac{\mbox{NJ}}{B_8}$}},
    every head row/.style={before row=\toprule,after row=\midrule
        \multirow{7}{*}{Influenza},
    },
    every nth row={7}{before row=\midrule \multirow{7}{*}{Various},},
    every nth row={14}{before row=\midrule \multirow{7}{*}{Hepatitis},},
    every nth row={21}{before row=\midrule \multirow{7}{*}{Chromosomes},},
    every last row/.style={after row=\bottomrule},
    %hide the unsuccessful algorithms
    skip rows between index={7}{9},
    skip rows between index={10}{11},
    skip rows between index={19}{20},
    skip rows between index={21}{22},
    skip rows between index={23}{24},
    skip rows between index={31}{34},
    skip rows between index={43}{45},
    skip rows between index={46}{47},
    % hide worse of MES and MESS
    skip rows between index={2}{3},
    skip rows between index={14}{15},
    skip rows between index={27}{28},
    skip rows between index={39}{40},
    % hide worse of MESSG and MESSM
    skip rows between index={5}{6},
    skip rows between index={16}{17},
    skip rows between index={29}{30},
    skip rows between index={41}{42},
  ]\alldata
\end{table*}
\end{landscape}

\begin{figure*}[h!]
\centering
  \begin{tabular}{cc}
  \begin{tikzpicture}
    \begin{axis}[width=0.50\textwidth,
      legend columns=4,
      legend to name=legend:subplotlegend,
      title={Influenza},
      ymin=0,
      ymax=1,
      ]
      %
      %\addlegendentry{reference};
      \addlegendentry{$\max(|R_A|,|R_B|)$};
      \addlegendentry{$\setdist_{\textsf{MES}}$};
      \addlegendentry{$\setdist_{\textsf{MESS}}$};
      \addlegendentry{$\setdist_{\textsf{MESSG}}$};
      \addlegendentry{$\setdist_{\textsf{MESSGM}}, \theta'=0.35$};
      \addlegendentry{$\setdist_{\textsf{MESSGq}}$};
      \addlegendentry{ABySS};
      \addlegendentry{Edena};
      \addlegendentry{SSAKE};
      \addlegendentry{SPAdes};
      \addlegendentry{Velvet};
      \foreach \method in {maxSize, MongeElkan, MongeElkanScale, MarginGaps, threshold, tripletsMG, abyss, edena, ssake, spades, velvet}
          \addplot table [x={k}, y={\method}] {\influenzaFMNJ};
    \end{axis}
  \end{tikzpicture}
  &
  \begin{tikzpicture}
    \begin{axis}[width=0.50\textwidth,title={Various},ymin=0,ymax=1,]
      \foreach \method in {maxSize, MongeElkan, MongeElkanScale, MarginGaps, threshold, tripletsMG, abyss, edena, ssake, spades, velvet}
          \addplot table [x={k}, y={\method}] {\variousFMNJ};
    \end{axis}
  \end{tikzpicture}
  \\
  \begin{tikzpicture}
    \begin{axis}[width=0.50\textwidth,title={Hepatitis},ymin=0,ymax=1,]
      \foreach \method in {maxSize, sMongeElkan, sMongeElkanScale, sMarginGaps, sthreshold, tripletsMG, abyss, edena, ssake, spades, velvet}
          \addplot table [x={k}, y={\method}] {\hepatitisFMNJ};
    \end{axis}
  \end{tikzpicture}
  &
  \begin{tikzpicture}
    \begin{axis}[width=0.50\textwidth,title={Chromosomes},ymin=0,ymax=1,]
      \foreach \method in {maxSize, MongeElkan, MongeElkanScale, MarginGaps, threshold, tripletsMG, abyss, edena, ssake, spades, velvet}
          \addplot table [x={k}, y={\method}] {\chromsFMNJ};
    \end{axis}
  \end{tikzpicture}
  \end{tabular}
\ref{legend:subplotlegend}
\caption{Plots of Fowlkes-Mallows index $B_k$ versus $k$. The index compares trees generated by the neighbor-joining algorithm. The tree is compared with the tree generated from the original sequences. If all values are equal to $1$, the structures of the trees are the same.}
\label{fig:fmindex}
\end{figure*}
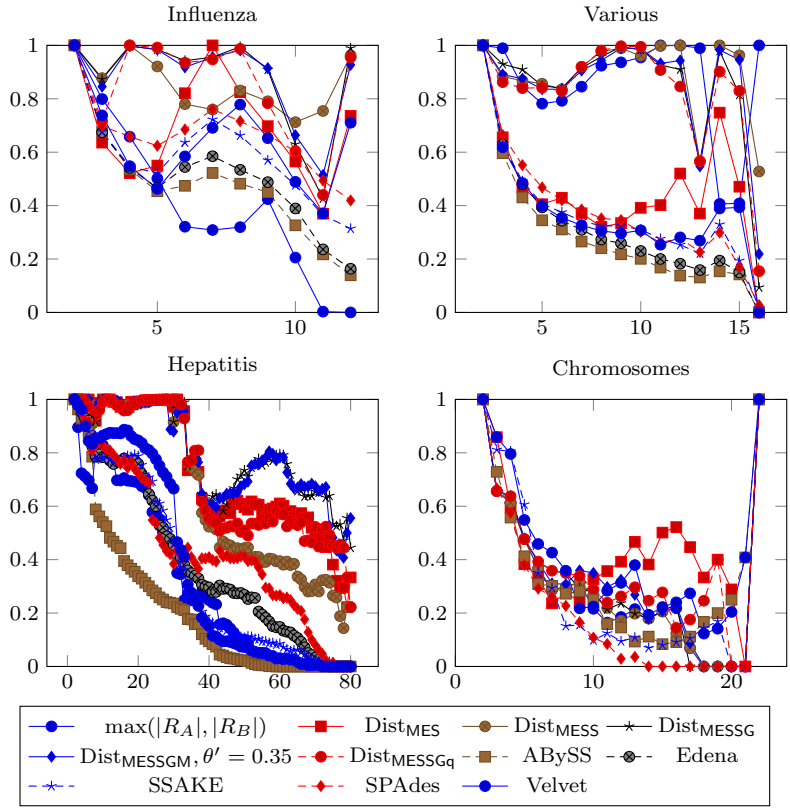

The {\bf purpose} of the experiments is to compare different methods for estimating the Levenhstein distance $\dist(A,B)$ for various real DNA sequences $A, B$ from their read sets $R_A, R_B$. The methods include (i) our newly proposed distances ({\scriptsize \textsf{MES}, \textsf{MESS}, \textsf{MESSG}, \textsf{MESSGM}, \textsf{MESSGq}}) applicable directly on $R_A, R_B$ and implemented in Java with maximum of shared code, (ii) the conventional method based on assembling estimates  $\hat{A}, \hat{B}$ of the original sequences $A,B$ using 5 common de-novo gene assemblers (ABySS \citep{abyss}, Edena \citep{edena}, SSAKE \citep{ssake}, SPAdes \citep{spades} and Velvet \citep{velvet}) and then estimating $\dist(A,B)$ as $\dist(\hat{A}, \hat{B})$, (iii) a trivial baseline method estimating $\dist(A,B)$  as $ \max\{|R_A|,|R_B|\}$. Our code was single-threaded. All the 5 assembly algorithms were configured with the default parameters and the current official C++ version was used.
When a result of an assembly procedure consisted of multiple contigs, we selected the two best matching contigs to represent the distance.

The evaluation {\bf criteria} consist of (i) the Pearson's correlation coefficient measuring the similarity of the distance matrix produced by the respective methods to the true distance matrix, (ii) The Fowlkes-Mallows index \citep{fmindex} measuring the similarity between the tree produced by a hierarchical clustering algorithm using the true distance matrix, and the tree induced from a distance matrix estimated by the respective method, (iii) runtime needed for assembly (if applicable), distance matrix calculation, and clustering time. For hierarchical clustering, we used the UPGMA algorithm \citep{upgma} and the neighbor-joining algorithm \citep{neighjoin}. The Fowlkes-Mallows index shows how much the resulting trees differ in structure. Both trees are first cut into $k$ clusters for $k = 2, 3, \ldots, n-1$. Then the clusterings are compared based on the number of common objects among each pair of clusters. By this procedure we obtain a set of values $B_k$ that shows how much the trees differ at various levels.

The testing {\bf data} contains four datasets. They are summarized in Table \ref{tab:datasets}. The ``influenza'' dataset\footnotemark
\footnotetext{\str{AF389115}, \str{AF389119}, \str{AY260942}, \str{AY260945}, \str{AY260949}, \str{AY260955}, \str{CY011131}, \str{CY011135}, \str{CY011143}, \str{HE584750}, \str{J02147}, \str{K00423} and outgroup \str{AM050555}. The genomes are available at \texttt{http://www.ebi.ac.uk/ena/data/view/<accession>}.}
contains 12 influenza virus genome sequences plus an outgroup sequence. The ``various'' dataset\footnotemark
\footnotetext{\str{AB073912}, \str{AB236320}, \str{AM050555}, \str{D13784}, \str{EU376394}, \str{FJ560719}, \str{GU076451}, \str{JN680353}, \str{JN998607}, \str{M14707}, \str{U06714}, \str{U46935}, \str{U66304}, \str{U81989}, \str{X05817}, \str{Y13051} and outgroup \str{AY884005}}
contains 17 genomes of different viruses. Furthermore, we used an independent third {\em training} dataset\footnotemark\ to tune the value of $\theta'$ (see Sect. \ref{ssec:threshold}). All the sequences were downloaded from the ENA repository \url{http://www.ebi.ac.uk/ena}.
\footnotetext{\str{CY011119}, \str{CY011127}, \str{CY011140}, \str{FJ966081}, \str{AF144300}, \str{AF144300}, \str{J02057}, \str{AJ437618}, \str{FR717138}, \str{FJ869909}, \str{L00163}, \str{KJ938716}, \str{KP202150}, \str{D00664}, \str{HM590588}, \str{KM874295}, $\alpha=4$, $l=40$}
Those two datasets contained artificial reads that were sampled under assumptions that reads are i.i.d. and that $\alpha$ and $l$ are constant. We sampled those two datasets several times with a high range of coverage and read length\footnote{$(\alpha, l) \in \{0.1, 0.3, 0.5, 0.7, 1, 1.5, 2, 2.5, 3, 4, 5, 7, 10, 15, 20, 30,\allowbreak 40,\allowbreak 50,\allowbreak 70,\allowbreak 100\} \times \{3, 5, 10, 15, 20, 25, 30, 40, 50, 70, 100, 150,\allowbreak 200,\allowbreak 500\}$} to obtain a representative information about algorithms behavior based on $\alpha$ and $l$. For averaging purposes we left out the three most outlying values of $\alpha$ and $l$.

\begin{figure}
  \centering
  \begin{tikzpicture}
    \begin{axis}[width=0.5\textwidth,
      %height=0.30\textwidth,
      legend columns=4,
      legend to name=legend:subplotlegendCorr,
      title={Influenza},
      ymin=0.35,
      ymax=1,
      xmin=0.1,
      xmax=100,
      xmode=log,
      log ticks with fixed point,
      ]
      %
      %\addlegendentry{reference};
      \addlegendentry{$\max(|R_A|,|R_B|)$};
      \addlegendentry{$\setdist_{\textsf{MES}}$};
      \addlegendentry{$\setdist_{\textsf{MESS}}$};
      \addlegendentry{$\setdist_{\textsf{MESSG}}$};
      \addlegendentry{$\setdist_{\textsf{MESSGM}}$};
      \addlegendentry{$\setdist_{\textsf{MESSGq}}$};
      \addlegendentry{ABySS};
      \addlegendentry{Edena};
      \addlegendentry{SSAKE};
      \addlegendentry{SPAdes};
      \addlegendentry{Velvet};
      \foreach \method in {maxSize, MongeElkan, MongeElkanScale, MarginGaps, threshold, tripletsMG, abyss, edena, ssake, spades, velvet}
          \addplot table [x={coverage}, y={\method}] {\influenzaCoverageSeries};
    \end{axis}
  \end{tikzpicture}
  \begin{tikzpicture}
    \begin{axis}[
      width=0.5\textwidth,
      %height=0.30\textwidth,
      title={Various},
      ymin=0.2,
      ymax=1,
      xmin=0.1,
      xmax=100,
      xmode=log,
      log ticks with fixed point,
      ]
      \foreach \method in {maxSize, MongeElkan, MongeElkanScale, MarginGaps, threshold, tripletsMG, abyss, edena, ssake, spades, velvet}
          \addplot table [x={coverage}, y={\method}] {\variousCoverageSeries};
    \end{axis}
  \end{tikzpicture}
\ref{legend:subplotlegendCorr}
\caption{Plot of average Pearson's correlation coefficient for several choices of coverage values.}
\label{fig:corseries}
\end{figure}
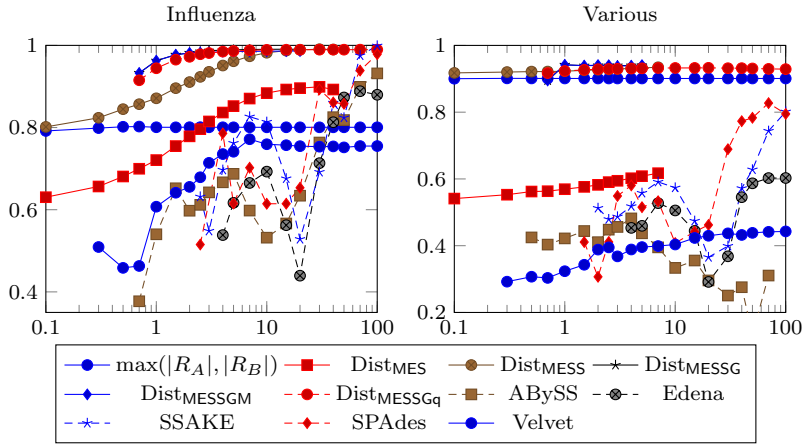

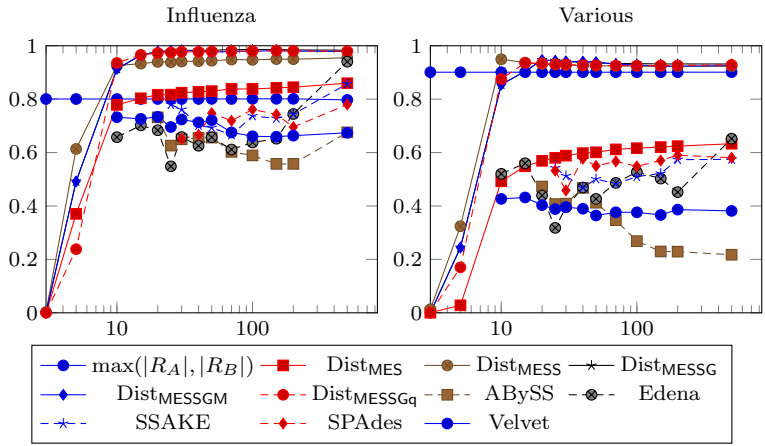
\begin{figure}
\centering
  \begin{tikzpicture}
    \begin{axis}[width=0.50\textwidth,
      %height=0.30\textwidth,
      legend columns=4,
      legend to name=legend:subplotlegendRL,
      title={Influenza},
      ymin=0,
      ymax=1,
      xmin=3,
      xmode=log,
      log ticks with fixed point,
      ]
      %
      %\addlegendentry{reference};
      \addlegendentry{$\max(|R_A|,|R_B|)$};
      \addlegendentry{$\setdist_{\textsf{MES}}$};
      \addlegendentry{$\setdist_{\textsf{MESS}}$};
      \addlegendentry{$\setdist_{\textsf{MESSG}}$};
      \addlegendentry{$\setdist_{\textsf{MESSGM}}$};
      \addlegendentry{$\setdist_{\textsf{MESSGq}}$};
      \addlegendentry{ABySS};
      \addlegendentry{Edena};
      \addlegendentry{SSAKE};
      \addlegendentry{SPAdes};
      \addlegendentry{Velvet};
      \foreach \method in {maxSize, MongeElkan, MongeElkanScale, MarginGaps, threshold, tripletsMG, abyss, edena, ssake, spades, velvet}
          \addplot table [x={readLength}, y={\method}] {\influenzaRLSeries};
    \end{axis}
  \end{tikzpicture}
  \begin{tikzpicture}
    \begin{axis}[
      width=0.50\textwidth,
      %height=0.30\textwidth,
      title={Various},
      ymin=0,
      ymax=1,
      xmin=3,
      xmode=log,
      log ticks with fixed point,
      ]
      \foreach \method in {maxSize, MongeElkan, MongeElkanScale, MarginGaps, threshold, tripletsMG, abyss, edena, ssake, spades, velvet}
          \addplot table [x={readLength}, y={\method}] {\variousRLSeries};
    \end{axis}
  \end{tikzpicture}
\ref{legend:subplotlegendRL}
\caption{Plot of average Pearson's correlation coefficient for several choices of read length.}
\label{fig:rlseries}
\end{figure}

The ``hepatitis'' dataset contains 81 Hepatitis A segments. This time we used ART \citep{art} program to similate sequencing to obtain read data for $(\alpha, l) \in \{10, 30, 50\} \times \{30, 70, 100\}$. We sampled with higher coverage and in case of our methods we down-sampled to coverage of $2$.

The ``chromosomes'' dataset contains $23$ regions of human genome. The data were obtained from \emph{1000 Genomes Project} \citep{1kgenomes}. For each chromosome we selected one $20 \, \mathrm{kbp}$ long region and obtained reads that were sequenced from this region. We used the ensembl \citep{ensembl} reference human genome to calculate the reference distance matrix. In this dataset the average coverage per chromosome was $4.32$. However for particular chromosomes this value varied a lot.

In the preliminary {\bf tuning} experiment, the value $\theta' = 0.35$ achieved the best Pearson's correlation coefficient on the training dataset and we carried this value over to the testing experiments.

A time limit was applied on assembly step and distance matrix calculation. Whenever an algorithm did not finish in time or assembly step failed to produce any contig, we marked the attempt as unsuccessful and did not count it towards the average.

The main experimental {\bf results} are shown in Table \ref{tab:results}. The four partitions of the table correspond to the average results on the four datasets. The Pearson's correlation coefficient (column `corr.') demonstrates dominance of the {\scriptsize \textsf{MESSG}} and {\scriptsize \textsf{MESSGM}} methods (Sects. \ref{ssec:weakborder} and \ref{ssec:threshold}), which are the most developed versions of our approach. The {\scriptsize \textsf{MESSG}} differs from {\scriptsize \textsf{MESSGM}} only by not discarding poorly matching reads. This finding is generally supported also by the Fowlkes-Mallows index (last four columns) shown for two levels of trees learned by two methods. Fig. \ref{fig:fmindex} provides a more detailed insight into the Fowlkes-Mallows values graphically for all the tree levels.

Fig. \ref{fig:corseries} shows how accurate are different algorithms based on coverage. We see that our approaches produce good results for coverage around $2$. Therefore on ``hepatitis'' dataset we used down-sampling to $\alpha=2$. The assembly algorithms however require several times higher coverage to produce results of the same accuracy. Fig. \ref{fig:rlseries} shows that our method produces good estimates for shorter reads than the assembly methods.

Columns 4-5 of  Table \ref{tab:results} indicate that the exact variants ({\scriptsize \textsf{MES}}, {\scriptsize \textsf{MESS}}, {\scriptsize \textsf{MESSG}}, {\scriptsize \textsf{MESSGM}}) of our approach were systematically slower in terms of absolute runtime than the approaches based on sequence assembly, despite the NP-hard complexity of the latter task. However the approximated version ({\scriptsize \textsf{MESSGq}}, including sampling) of our approach did produce results in time, which is comparable to the assembly time. The cost for this runtime improvement was only a small decrease of accuracy. Moreover the {\scriptsize \textsf{MESSGq}} approach was faster than calculating the reference distance matrix on larger datasets. The numbers also show that our asymptotic complexity estimate in Sect. \ref{ssec:asymptotic} is generally correct: the ratio between the time spent on calculating the distances on one hand, and the runtime of the reference method on the other hand, is approximately $\alpha^{2}$.

The ``chromosomes'' dataset shows how our method is dependent on the assumption that estimates of the original sequences lengths are good. While method {\scriptsize \textsf{MES}} provides good estimates, the accuracy drops after applying the scaling from Sect. \ref{ssec:scale}. The coverage is not constant on all read bags and therefore the results are poor knowing that all the original sequences had the same length.

\section{Conclusion}

We have proposed and evaluated several variants of a method for estimating edit distances between sequences only from read-sets sampled from them. In experiments, our approach produced better estimates than a conventional approach based on first estimating the sequences themselves by applying assembly algorithms on the read-sets. 

From the experiments we see that the assembly algorithms require a higher coverage $\alpha$ and read length $l$ to produce estimates comparable to our approach. This fact provides a possibility to reduce the amount of wet-lab work, which is needed to apply machine learning methods based on similarity on read sets.

Furthermore, our approach offers many directions for technical improvements. For example, one may consider a {\em partial assembly} approach, in which sets of a few (up to a constant) reads would be pre-assembled and the Monge-Elkan distance would be applied on such partial assemblies. This view would open a `continuous' spectrum between our approach on one hand, and the conventional assembly-based approach, on which the optimal trade-off could be identified. 

\begin{acknowledgements}
Access to computing and storage facilities owned by parties and projects contributing to the National Grid Infrastructure MetaCentrum, provided under the programme "Projects of Large Research, Development, and Innovations Infrastructures" (CESNET LM2015042), is greatly appreciated.

\end{acknowledgements}

\bibliographystyle{plain}
\bibliography{files/references}

\end{document}